\documentclass[11pt]{article}

\usepackage[letterpaper,margin=0.78in]{geometry}
\usepackage[T1]{fontenc}
\usepackage{lmodern}
\usepackage{microtype}
\usepackage{amsmath,amssymb,amsthm,mathtools,bm}
\usepackage{graphicx}
\usepackage{booktabs,tabularx,array}
\usepackage{xcolor}
\usepackage{enumitem}
\usepackage{caption}
\usepackage{hyperref}
\usepackage{float}
\usepackage{placeins}
\usepackage{hyperref}
\usepackage{xurl}

\hypersetup{
  colorlinks=true,
  linkcolor=blue!55!black,
  citecolor=blue!55!black,
  urlcolor=blue!65!black
}
\graphicspath{{figures/}}
\setlist[itemize]{leftmargin=1.6em,itemsep=2pt,topsep=3pt}
\newtheorem{theorem}{Theorem}

\newtheorem{lemma}[theorem]{Lemma}
\newtheorem{corollary}[theorem]{Corollary}
\theoremstyle{definition}

\newcommand{\D}{D}
\newcommand{\TV}{D_{\mathrm{TV}}}
\newcommand{\Tr}{\operatorname{Tr}}
\newcommand{\cM}{C_{\mathrm{MERA}}}
\newcommand{\cU}{\overline C_{\mathrm{MERA}}}
\newcommand{\Ftop}{F_{\mathrm{top}}}
\newcommand{\ketbra}[2]{\lvert #1\rangle\!\langle #2\rvert}
\newcommand{\mean}[1]{\overline{#1}}

\newcolumntype{Y}{>{\raggedright\arraybackslash}X}

\title{\bfseries Terminal-Register Certification for Finite-Measurement Learning of Multiscale Quantum States}

\author{%
  Bhavin Makwana\textsuperscript{1},
  Kashyap Patel\textsuperscript{1},
  Manjunath Joshi\textsuperscript{1},
  and Jaideep Mulherkar\textsuperscript{2}\\[0.6em]
  \textsuperscript{1}Dhirubhai Ambani University, Gandhinagar, Gujarat, India\\
  \textsuperscript{2}Georgia Institute of Technology, Atlanta, Georgia, USA
}

\begin{document}
\maketitle

\begin{abstract}

Structured quantum-state learning not only depends on an expressive ansatz but also on an operational certificate that stays meaningful with finite measurements and imperfect implementation. We study pure one dimensional states learning by an inverse binary multiscale entanglement renormalization ansatz (MERA). In the learning procedure, the qubits removed during coarse graining are controlled coherently and measured together at the terminal register. In the case of a fixed analysis circuit, we establish that the exact identity $F=p_0F_{top}$ between reconstruction fidelity, conditional top state fidelity, and terminal all zero probability. This identity produces layerwise, global, and local discarded-weight certificates, coexisting final shot bounds, and observable error guarantees. We confirm that an ideal sequential and terminal measurement schedule delivers the same complete bit string distribution under matched causal operations, while normalized postselection can amplify perturbations inversely with prefix acceptance. A noise aware theorem introduces an individual calibrated total variation implementation budget to the finite shot certificate. The protocol is estimated on an open boundary transverse field Ising ground state. A frozen 8-qubit schedule using $560$ million simulated training measurements per run achieves fidelity above $0.99$ in all $60$ held-out runs, with a mean fidelity of $0.996886$. 1080 circuit-noise cells and 6480 confidence-coverage rows are covered by fixed-circuit robustness validation without a locked soundness violation. We then address architectural fairness at $n=16$ using three new studies. A 255-coordinate MERA-PM225 and MPS have statistically unresolved fidelity and energy differences at the largest exposure under approximately matched aggregate $SU(4)$-shot-gate exposure. Experiments demonstrate that MERA retains lower long-range connected-$ZZ$ error ($\Delta=-0.01458$, 97.5\% interval $[-0.02811,-0.00034]$) and lower half-chain-entropy error ($\Delta=-0.03044$, interval $[-0.04878,-0.01044]$). In a 120-run exact-gradient multistart diagnostic, MERA has higher fidelity in 58/60 paired restarts and lower long-range error in 60/60, although no run met the prespecified stationarity criterion. Finally, a causal cone-complete, parameter matched local circuit achieves $2.62\times$ greater aggregate gate exposure yet loses all 30 paired comparisons in fidelity, long-range error, energy, and entropy. The evidence supports sound terminal-register certification and a resource-qualified multiscale advantage, not universal MERA superiority or a certified global optimum.

\end{abstract}

\noindent\textbf{Keywords:} quantum state tomography; MERA; tensor networks; finite measurements; certification; postselection

\section{Introduction}

Quantum learning is a basic task in quantum information, learning an unknown quantum state. Quantum tomography becomes expensive as the number of qubits grows. An arbitrary n-qubit density operator retains exponentially many real degrees of freedom, and likewise, general reconstruction protocols require exponentially many measurement outcomes and powerful classical post-processing. When the state is known to be nearly pure or low rank, low rank recovery can reduce this burden \cite{gross2010quantum}, and direct fidelity analysis and classical shadow protocols efficiently target selected properties without constructing a generative state model \cite{flammia2011direct, huang2020predicting}. These methods explain an important discrepancy: rather than reconstructing a state that can itself be queried, simulated, or prepared, evaluating a family of observables can be much easier. Additional structure is provided by many body physics. Tensor network can accurately describe a ground state of local one dimensional Hamiltonians, with the limitation of exponentially growing parameter size along with system size. Matrix product state (MPS) encodes one dimensional area law states through a finite dimensional bond \cite{perez2006matrix}. MPS structure has encouraged efficient and certifiable tomography schemes and experimental reconstructions \cite{cramer2010efficient, lanyon2017efficient}. The multiscale entanglement renormalization ansatz (MERA) augments isometries with local disentanglers and organizes them hierarchically, delivering logarithmic depth between distant sites and representation of scale dependent correlaions \cite{vidal2007entanglement, vidal2008class}. Thus, MERA is an attractive inductive bias for quantum state learning. However, it also introduces a certification question that does not emerge in a form that is different from a sequential chain. 

This paper not only works on optimization of MERA, but it also aims to connect a finite set of experimentally obtainable bit strings to a valid statement about the quality of the reconstructed state. The training loss may be a sum of local quantities, may depend on postselected states,
or may be evaluated with finite-shot error. Moreover, the state reconstructed by reversing the learned circuit includes both the contained top state and all qubits dismissed during coarse graining. A useful certificate should account for these components in one consistent probability model, expose where errors occur across scales, and remain sound when the executed terminal distribution varies slightly from the ideal one.

Existing MERA-tomography methods provide the immediate starting point. Landon-Cardinal and Poulin proposed a practical layer-by-layer learning method in which disentangled subsystems are measured, conditionally retained, and passed to the next scale \cite{landon2012practical}. Lee and Landon-Cardinal developed a variational version for critical one-dimensional systems and studied experimentally accessible truncation bounds \cite{lee2015practical}. These works show that multiscale structure can make otherwise intractable reconstruction problems operational. At the same time, a sequential presentation naturally uses normalized conditional states. If the probability of reaching a later layer is small, a modest perturbation of the unnormalized branch can become large after division by the prefix-acceptance probability. This observation does not invalidate sequential tomography, nor does deferred measurement create a measurement-complexity separation by itself. It does motivate a formulation in which the exact joint event, the conditional top state, and the statistical allowances are written separately.

We introduce such a formulation by retaining every qubit discarded by the binary MERA analysis as part of a terminal register. The analysis circuit is executed coherently from the physical layer up. Each shot produces the complete discarded bit string together with a retained top qubit, and the event that all discarded bits are zero is estimated directly. In the case of a fixed analysis circuit, the overlap of inverse circuit reconstruction and input factorizes precisely into this all zero probability, and the result is the fidelity of the conditional top state estimate. Once the registers are specified, identity is an element, which results in several practical outcomes: It holds layerwise, global, and local certificates from the same records, also it bypasses any independence hypotheses between discarded qubits; and it allows finite shot and implementation terms to be added transparently.

The terminal viewpoint should not be confused with a claim that one physical circuit is intrinsically less noisy than another. Under identical causal operations and ideal measurements, sequential measurements and terminal measurements generate the same complete bit-string distribution. We prove this equivalence explicitly. The robustness distinction lies elsewhere: normalized postselection is a nonlinear map, and its worst-case sensitivity scales inversely with the probability of the accepted prefix. In contrast, the terminal certificate is expressed as an unconditional event probability. When a separately calibrated total-variation budget bounds the deviation between ideal and implemented terminal distributions, that budget enters the fidelity bound additively. Thus, the result is a soundness statement under a stated implementation model, not an assertion that noise disappears.

Certification is not enough to demonstrate that the underlying representation is useful. We therefore experimented with ground states of the open-boundary transverse-field Ising model (TFIM), using pure states and binary bond dimension $\chi = 2$. Hyperparameters are chosen only on calibration fields and seeds; the resulting schedule is frozen before held-out fields are opened. Reconstruction quality is estimated not only through oracle fidelity, but also through held out observables, entanglement entropy, energy density error, and terminal certificate. We demonstrate the comparison of MERA with TTN, MPS, and shallow local circuits. Tensor network architecture can have a different number of parameter count and number of two qubit operations, independent coordinate matching, and simultaneous gates.

Our contributions are:
\begin{itemize}
    \item We prove the exact reconstruction identity $F = p_0F_{top}$ and derive the hierarchy $1 - p0 \leq C_{MERA} \leq \overline{C}_{MERA}$, linking one joint terminal event to layerwise and local discarded-weight diagnostics without considering independence.
    \item We convert the identical identities into coexisting finite shot, top temography, and bounded observable guarantees. All confidence distributions and union bounds are stated explicitly.
    \item We demonstrated that matched sequential and terminal schedules have the same complete ideal record allocation, while normalized prefix postselection shows inescapable inverse acceptance sensitivity.
    \item We present noise aware certificate in which a calibrated total variation implementation budget arises additively,  purely dividing statistical uncertainty from implementation mismatch.
    \item We validate the full certificate numerically and report a frozen $60$ run, 8-qubit TFIM evaluation in which every held out reconstruction surpasses fidelity $0.99$.
    \item We separate three fairness questions at $n=16$: matched aggregate gate exposure against MPS, exact-gradient multistart behavior, and a causal-cone-complete local control. At the largest matched exposure, the energy difference and fidelity between MERA-PM225 and MPS are statistically undertermined.
\end{itemize}

\section{Related Work}

Quantum tomography aims for a complete classical description of an unknown quantum state, and the cost is exponential without structural hypotheses. Compressed sensing methods influence low rank \cite{gross2010quantum}, while local reduction and scalable likelihood methods influence positivity or consistency constraints \cite{baumgratz2013scalable}. Direct fidelity estimation instead certifies overlap with a selected target using importance sampled Pauli expectations \cite{flammia2011direct}. Classical shadows assemble a reusable classical record for predicting many observables, with performance determined by the chosen measurement ensemble and observable family \cite{huang2020predicting}. These approaches focus on complementary goals: neither property prediction nor target specific fidelity estimation automatically constructs the structured inverse circuit studied here.

The tensor network tomography model focuses on using the entanglement structure to reach a solution. Cramer et al. established an efficient, certifiable architecture to approximate the state using MPS \cite{cramer2010efficient}. Lanyon et al. experimented on MPS tomography in tapped ion many body systems \cite{lanyon2017efficient}.
Neural network tomography takes you to another generative route, training model, and phases from expectation value data without imposing a detailed tensor network geometry \cite{torlai2018neural}. In case of critical one dimensional systems, Landon-Cardinal and Poulin designed sequential MERA learning \cite{landon2012practical}. Lee and Landon-Cardinal suggested a practical variational refinement with truncation error analysis \cite{lee2015practical}. Our work is most comparable to these two MERA methods, but with a different certification question: how can one joint terminal record authenticate the inverse-circuit reconstruction with finite shots and a calibrated implementation mismatch?

Table \ref{tab:related-work} outlines this positioning. The declared distinction is not that terminal measurement changes the ideal record distribution; our equivalence theorem rules out that interpretation. The contribution is the combination of an exact joint-event identity, simultaneous multiscale diagnostics, explicit finite-sample confidence, top-state accounting, and a noise-aware additive correction, followed by preregistered resourc matched architecture tests.

\begin{table}[t]
\centering
\caption{Positioning relative to representative state-learning and certification methods. Full model means that the output specifies a generative state or circuit rather than only selected properties.}
\label{tab:related-work}
\scriptsize
\setlength{\tabcolsep}{3.2pt}
\begin{tabularx}{\linewidth}{@{}>{\raggedright\arraybackslash}p{0.18\linewidth}>{\raggedright\arraybackslash}p{0.15\linewidth}>{\raggedright\arraybackslash}p{0.14\linewidth}>{\raggedright\arraybackslash}p{0.18\linewidth}Y@{}}
\toprule
Method & Structural prior & Full model & Certificate or guarantee & Distinction from this work \\
\midrule
Compressed sensing~\cite{gross2010quantum} & Low rank & Yes & Recovery from suitable random measurements & No multiscale inverse circuit or terminal register \\
Direct fidelity estimation~\cite{flammia2011direct} & Known target & No & Target-overlap estimate & Certifies a supplied target rather than learning MERA \\
Classical shadows~\cite{huang2020predicting} & Observable-dependent & No & Many-property prediction & Does not by itself reconstruct the inverse circuit \\
MPS tomography~\cite{cramer2010efficient, lanyon2017efficient} & 1D low entanglement & Yes & Parent-Hamiltonian/local consistency bounds & Chain geometry rather than multiscale discarded registers \\
Neural tomography~\cite{torlai2018neural} & Neural ansatz & Yes & Likelihood-based reconstruction & No terminal discarded-weight certificate \\
Sequential MERA~\cite{landon2012practical,lee2015practical} & Multiscale entanglement & Yes & Layerwise truncation control & Uses sequential conditional states and proxies \\
Variational/disentangling
circuits\cite{yao2023quantum, liu2020variational}&Circuit ansatz&Yes&Objective-dependent reconstruction&No multiscale terminal-register theorem\\

This work & Binary MERA/inverse circuit & Yes & Exact joint identity; finite-shot and noise-aware bounds & One complete terminal record with global, layer, and local diagnostics \\
\bottomrule
\end{tabularx}
\end{table}

\section{Problem Setting}

\subsection{Target Family}

We work with normalized pure states on an open chain of $n$ qubits. Our main targets are ground states of the open-boundary transverse-field Ising Hamiltonian
\begin{equation}
H(J,h)=-J\sum_{i=1}^{n-1}Z_iZ_{i+1}-h\sum_{i=1}^{n}X_i .
\end{equation}

The critical region sits near $h/J=1$. For our present small and intermediate size simulations, exact diagonalization delivers us the target states and oracle observables. Instead of direct access to target amplitudes, training of the learner relies on measurement estimates, except in specifically labeled noiseless diagnostics. Throughout, 
\begin{equation}
F(\lvert a\rangle,\lvert b\rangle)=|\langle a\vert b\rangle|^2
\end{equation}
denotes squared fidelity.

\subsection{Binary MERA analysis circuit}

We use an open-boundary binary MERA with bond dimension $\chi=2$ as our main architecture. Each layer starts with two-qubit disentanglers across block boundaries, then applies two-qubit unitary dilations of the isometries within each block. One qubit per isometry block is thrown away, and later layers never touch wires that were already discarded. For $n=8$, the analysis goes $8\to4\to2\to1$, so we end up with seven discarded qubits and one top qubit.

A disentangler is an $SU(4)$ element with 15 real parameter values. An isometry in MERA is $w:\mathbb C^2\to\mathbb C^2\otimes\mathbb C^2$, which has 12 effective degrees of freedom. The primary model of MERA contains independent tensors, and we tried shared tensors as an alternative. MPS baseline for MERA is a sequential nearest-neighbor inverse circuit with the same local $SU(4)$ block.

Let $D=\bigsqcup_{\ell=0}^{L-1}D_\ell$ be the disjoint`  union of discarded register and $T$ the retained top register. Given an input $\lvert\psi\rangle$ and analysis unitary $U_{\theta}$, we write the output as

\begin{equation}
\lvert\Psi_\theta\rangle=U_\theta\lvert\psi\rangle\in\mathcal H_T\otimes\mathcal H_D .
\label{eq:3}
\end{equation}

In TTN architecture, there is no disentangler. To compute the MPS comparator, we use a sequential nearest-neighbor inverse circuit. At $n=16$, MPS and TTN have 15 blocks each and 225 coordinates, while full MERA has physical $SU(4)$ blocks with 390 independent coordinates. Our parameter matched MERA shares parameters across selected tensors along with keeping all 225 coordinates, so it ends up with 15 independent $SU(4)$ groups and 225 coordinates. Finally, Local-225 and Local-390 serve as shallow nearest-neighbor brickwork controls, containing 15 and 26 independent blocks, respectively.

\subsection{Reconstruction}

Define the all-zero projector
\begin{equation}
\Pi=I_T\otimes\ketbra{0^D}{0^D},
\qquad
p_0=\langle\Psi_\theta\rvert\Pi\lvert\Psi_\theta\rangle.
\label{eq:p0}
\end{equation}
For $p_0>0$, the normalized conditional top state is
\begin{equation}
\lvert\phi_T\rangle=
\frac{(I_T\otimes\langle0^D\rvert)\lvert\Psi_\theta\rangle}{\sqrt{p_0}}.
\label{eq:conditional-top}
\end{equation}
Given a normalized top estimate $\lvert\widehat\phi_T\rangle$, inverse-circuit reconstruction returns
\begin{equation}
\lvert\widehat\psi\rangle=
U_\theta^\dagger
(\lvert\widehat\phi_T\rangle_T\otimes\lvert0^D\rangle_D).
\label{eq:reconstruction}
\end{equation}

\section{Terminal-Register Certificates}

\subsection{Exact certificate hierarchy}

For discarded layer $D_\ell$, let
\begin{equation}
\Pi_\ell=I_{\overline{D_\ell}}\otimes
\ketbra{0^{D_\ell}}{0^{D_\ell}},
\qquad
q_\ell=\langle\Psi_\theta\rvert(I-\Pi_\ell)\lvert\Psi_\theta\rangle.
\end{equation}
For its $j$th bit, define the unconditional discarded weight
$\eta_{\ell,j}=\Pr[D_{\ell,j}=1]$. The layer and local certificates are
\begin{equation}
\cM=\sum_{\ell=0}^{L-1}q_\ell,
\qquad
\cU=\sum_{\ell=0}^{L-1}\sum_{j\in D_\ell}\eta_{\ell,j}.
\label{eq:certificates}
\end{equation}

\begin{theorem}[Exact factorization and multiscale certificate]
\label{thm:exact}
For every pure input, unitary analysis circuit, and normalized pure top estimate,
\begin{equation}
F(\lvert\psi\rangle,\lvert\widehat\psi\rangle)
=p_0\Ftop,
\qquad
\Ftop=|\langle\phi_T\vert\widehat\phi_T\rangle|^2.
\label{eq:factorization}
\end{equation}
Moreover,
\begin{equation}
1-p_0\leq\cM\leq\cU,
\label{eq:hierarchy}
\end{equation}
and
\begin{equation}
1-F\leq
\min\{1,\cM+1-\Ftop\}
\leq
\min\{1,\cU+1-\Ftop\}.
\label{eq:main-bound}
\end{equation}
\end{theorem}

\begin{proof}
The analyzed state has the orthogonal decomposition
\begin{equation}
\lvert\Psi_\theta\rangle=
\sqrt{p_0}\lvert\phi_T\rangle\lvert0^D\rangle+
\sqrt{1-p_0}\lvert\xi_\perp\rangle,
\qquad \Pi\lvert\xi_\perp\rangle=0.
\end{equation}
Using Eqs.~\eqref{eq:reconstruction} and \eqref{eq:3}, unitary invariance gives
\begin{equation}
|\langle\widehat\psi\vert\psi\rangle|^2
=|\langle\widehat\phi_T,0^D\vert\Psi_\theta\rangle|^2
=p_0|\langle\widehat\phi_T\vert\phi_T\rangle|^2.
\end{equation}
Let $A_\ell$ be the event that $D_\ell$ is not all zero. The global failure event is $A=\bigcup_\ell A_\ell$, so
$1-p_0\leq\sum_\ell\Pr(A_\ell)=\cM$. Within a layer,
$A_\ell=\bigcup_{j\in D_\ell}\{D_{\ell,j}=1\}$, hence
$q_\ell\leq\sum_j\eta_{\ell,j}$ and $\cM\leq\cU$. Finally,
\begin{equation}
1-p_0\Ftop=(1-p_0)+p_0(1-\Ftop)
\leq(1-p_0)+(1-\Ftop),
\label{eq:14}
\end{equation}
which proves Eq.~\eqref{eq:main-bound}.
\end{proof}

The $q_\ell$ are unconditional marginals of one terminal distribution. In general,
$p_0\neq\prod_\ell(1-q_\ell)$; no independence assumption is used.

\begin{corollary}[Observable error]
For any Hermitian observable $O$,
\begin{equation}
\big|\Tr[O(\psi-\widehat\psi)]\big|
\leq2\|O\|_\infty
\sqrt{\min\{1,\cM+1-\Ftop\}},
\label{eq:observable-bound}
\end{equation}
and the same statement holds with $\cU$ in place of $\cM$.
\end{corollary}

\begin{proof}
For pure states, $\D(\psi,\widehat\psi)=\sqrt{1-F}$. Theorem~\ref{thm:exact} bounds $1-F$, and Holder's inequality gives
$|\Tr[O(\rho-\sigma)]|\leq\|O\|_\infty\|\rho-\sigma\|_1
=2\|O\|_\infty\D(\rho,\sigma)$.
\end{proof}

\subsection{Finite-measurement confidence}

Suppose each of the $N$ independent terminal executions returns the full discarded bit string. Consider a covering family of $K$ Bernoulli failure indicators, with true means $\mu_k$ and empirical means $\widehat\mu_k$. In the global family we take $K=1$ and mean $1-p_0$; the layer family has $K=L$ and sum $\cM$; the local-bit family has $K=|D|$ and sum $\cU$.

\begin{theorem}[Simultaneous one-sided finite-shot certificate]
\label{thm:finite}
For $\delta_D\in(0,1)$, set
\begin{equation}
b_K(N,\delta_D)=
\min\left\{1,
\sum_{k=1}^{K}\widehat\mu_k+
K\sqrt{\frac{\log(K/\delta_D)}{2N}}
\right\}.
\label{eq:finite-family}
\end{equation}
Then, with probability at least $1-\delta_D$,
\begin{equation}
1-p_0\leq b_K(N,\delta_D).
\end{equation}
\end{theorem}

\begin{proof}
Apply one-sided Hoeffding to each Bernoulli indicator,
$\Pr[\mu_k>\widehat\mu_k+t]\leq e^{-2Nt^2}$~\cite{hoeffding1963probability}. Take
$t=\sqrt{\log(K/\delta_D)/(2N)}$ and union bound over the $K$ indicators. Then with probability at least $1-\delta_D$, every $\mu_k\leq\widehat\mu_k+t$. Since each family covers the global failure event, we get
$1-p_0\leq\sum_k\mu_k\leq\sum_k\widehat\mu_k+Kt$.
\end{proof}

For the direct global estimator, $\widehat\mu_1=1-\widehat p_0$ and the inflation term is
$\sqrt{\log(1/\delta_D)/(2N)}$. The layer and local bounds give more diagnostic detail, but they are statistically looser.

When a single top qubit is retained, accepted shots are measured in the $X$, $Y$, and $Z$ bases. Let $K_b$ be the accepted count in basis $b$, set $K_{\min}=\min_bK_b$, and let the returned pure state point along the empirical Bloch vector.

\begin{lemma}[Finite-shot top term]
\label{lem:top}
Conditioned on positive accepted counts, with probability at least $1-\delta_T$,
\begin{equation}
1-\Ftop\leq
\varepsilon_{\mathrm{top}}:=
\min\left\{1,
\frac{6\log(6/\delta_T)}{K_{\min}}
\right\}.
\label{eq:top-bound}
\end{equation}
\end{lemma}

\begin{proof}
Since the outcomes lie in $[-1,1]$, two-sided Hoeffding gives
$\Pr(|\widetilde r_b-r_b|>a)\leq2e^{-K_ba^2/2}$. With
$a=\sqrt{2\log(6/\delta_T)/K_{\min}}$, a union bound over the three bases yields
$\|\widetilde r-r\|_2\leq\sqrt3a$. If $s=\widetilde r/\|\widetilde r\|_2$ is the returned unit vector, then
$\|s-r\|_2\leq2\|\widetilde r-r\|_2$. For pure qubits,
$1-\Ftop=\|s-r\|_2^2/4\leq3a^2$. Whenever the right-hand side exceeds one, we just fall back on the trivial infidelity bound.
\end{proof}

Putting Theorem~\ref{thm:finite} and Lemma~\ref{lem:top} together with a union bound gives
\begin{equation}
1-F\leq
\min\{1,b_K(N,\delta_D)+\varepsilon_{\mathrm{top}}\}
\label{eq:fully-finite}
\end{equation}
with probability at least $1-\delta_D-\delta_T$.

\section{Sequential Measurement, Postselection, and Noise}

\subsection{Exact equivalence of ideal schedules}

Order discarded registers as $D_1,\ldots,D_M$. Let $V_i$ denote the unitary at stage $i$ and let $P_i^{b_i}$ project $D_i$ onto outcome $b_i$. We impose the \emph{causal-discard condition}
\begin{equation}
[V_j,P_i^{b_i}]=0\qquad\text{for all }j>i,
\label{eq:causal-discard}
\end{equation}
and exclude outcome-dependent feed-forward.

\begin{theorem}[Deferred discarded-register measurement]
\label{thm:equivalence}
Under Eq.~\eqref{eq:causal-discard}, sequential measurement and joint terminal measurement generate the same complete-record distribution:
\begin{equation}
P_{\mathrm{seq}}(\bm b)=P_{\mathrm{term}}(\bm b).
\label{eq:distribution-equivalence}
\end{equation}
Consequently,
\begin{equation}
p_0=\prod_{i=1}^{M}Pr(B_i=0\mid B_{<i}=0).
\label{eq:conditional-product}
\end{equation}
\end{theorem}

\begin{proof}
For record $\bm b$, the sequential Kraus operator is
\begin{equation}
K_{\bm b}^{\mathrm{seq}}=
P_M^{b_M}V_M\cdots P_2^{b_2}V_2P_1^{b_1}V_1.
\end{equation}
Earlier projectors commute through later unitaries by Eq.~\eqref{eq:causal-discard}; projectors on distinct discarded registers commute. Therefore
\begin{equation}
K_{\bm b}^{\mathrm{seq}}=
\left(\prod_{i=1}^{M}P_i^{b_i}\right)V_M\cdots V_1
=K_{\bm b}^{\mathrm{term}}.
\end{equation}
The branch probabilities $\Tr(K_{\bm b}\rho K_{\bm b}^\dagger)$ are identical. Equation~\eqref{eq:conditional-product} is the classical chain rule for the all-zero event.
\end{proof}

From $N$ complete records, let $N_i$ count records whose first $i$ bits are zero and set $N_0=N$. Whenever $N_{i-1}>0$,
\begin{equation}
\prod_{i=1}^{M}\frac{N_i}{N_{i-1}}=\frac{N_M}{N}=\widehat p_0.
\end{equation}
Thus identical complete records give identical sequential-product and terminal estimators. We claim no universal confidence or shot-complexity advantage for estimating the same $p_0$ from the same data.

\subsection{Instability of normalized postselection}

For density operators $\rho,\sigma$, use trace distance
$\D(\rho,\sigma)=\tfrac12\|\rho-\sigma\|_1$. For a common projector $P$, define
\begin{equation}
p=\Tr(P\rho),\quad q=\Tr(P\sigma),\quad
\rho_P=\frac{P\rho P}{p},\quad
\sigma_P=\frac{P\sigma P}{q},
\end{equation}
when $p,q>0$.

\begin{theorem}[Inverse-acceptance postselection instability]
\label{thm:postselection}
If $\D(\rho,\sigma)\leq\tau$ and $p,q>0$, then
\begin{equation}
\D(\rho_P,\sigma_P)
\leq
\min\left\{1,
\frac{\tau+\tfrac12|p-q|}{\max\{p,q\}}
\right\}
\leq
\min\left\{1,
\frac{3\tau}{2\max\{p,q\}}
\right\}.
\label{eq:postselection-bound}
\end{equation}
The inverse-acceptance dependence is necessary in general.
\end{theorem}

\begin{proof}
Let $A=P\rho P$ and $B=P\sigma P$. Projector compression of a Hermitian operator does not increase its trace norm, so
$\tfrac12\|A-B\|_1\leq\tau$. Using denominator $p$,
\begin{align}
\D(A/p,B/q)
&\leq\frac{\|A-B\|_1}{2p}
+\frac12\|B\|_1\left|\frac1p-\frac1q\right|\\
&\leq\frac{\tau+\tfrac12|p-q|}{p}.
\end{align}
Interchanging $\rho$ and $\sigma$ gives the same numerator divided by $q$; choosing the smaller bound produces $\max\{p,q\}$ in the denominator. Contractivity under the binary measurement $\{P,I-P\}$ gives $|p-q|\leq\tau$, proving Eq.~\eqref{eq:postselection-bound}.

For necessity, take orthogonal $\lvert0\rangle,\lvert1\rangle,\lvert2\rangle$, let
$P=\ketbra0 0+\ketbra1 1$, and choose
\begin{align}
\rho&=p\ketbra0 0+(1-p)\ketbra2 2,\\
\sigma&=(p-\tau)\ketbra0 0+\tau\ketbra1 1+(1-p)\ketbra2 2,
\end{align}
where $0<\tau\leq p$. Both acceptance probabilities equal $p$,
$\D(\rho,\sigma)=\tau$, and $\D(\rho_P,\sigma_P)=\tau/p$.
\end{proof}

If a successful prefix has probability $\alpha$, the expected raw preparations required for $m$ accepted conditional samples are
\begin{equation}
\mathbb E[N_{\mathrm{raw}}]=\frac{m}{\alpha}.
\label{eq:survival-mean}
\end{equation}
A sufficient fixed budget for at least $m$ accepted samples with probability $1-\delta$ is
\begin{equation}
N_{\mathrm{raw}}\geq
\frac{\max\{2m,8\log(1/\delta)\}}{\alpha}.
\label{eq:survival-budget}
\end{equation}
Indeed, if $X\sim\mathrm{Bin}(N_{\mathrm{raw}},\alpha)$, the conditions imply
$m\leq\mathbb E[X]/2$ and the Chernoff bound \cite{chernoff1952} gives
$\Pr(X<m)\leq e^{-N_{\mathrm{raw}}\alpha/8}\leq\delta$.

\subsection{Noise-aware terminal certificate}

Let $P_{\mathrm{id}}$ be the ideal discarded-register distribution and
$P_{\mathrm{obs}}$ the implemented distribution. Assume a separately calibrated budget
\begin{equation}
\TV(P_{\mathrm{id}},P_{\mathrm{obs}})\leq\tau_D,
\qquad
\TV(P,Q)=\tfrac12\sum_x|P(x)-Q(x)|.
\label{eq:tv-budget}
\end{equation}

\begin{theorem}[Implementation-aware finite-shot certificate]
\label{thm:noise-aware}
Let $\widehat p_0$ be the all-zero frequency from $N$ independent samples of $P_{\mathrm{obs}}$. If an event of probability at least $1-\delta_T$ guarantees
$1-\Ftop\leq\varepsilon_{\mathrm{top}}$, then, with probability at least
$1-\delta_D-\delta_T$,
\begin{equation}
1-F\leq
\min\left\{1,
1-\widehat p_0+
\sqrt{\frac{\log(1/\delta_D)}{2N}}+
\tau_D+\varepsilon_{\mathrm{top}}
\right\}.
\label{eq:noise-aware}
\end{equation}
\end{theorem}

\begin{proof}
Total variation controls every event, hence
$1-p_0\leq1-\widetilde p_0+\tau_D$, where
$\widetilde p_0=P_{\mathrm{obs}}(0^M)$. One-sided Hoeffding gives \cite{hoeffding1963probability}
\begin{equation}
1-\widetilde p_0\leq1-\widehat p_0+
\sqrt{\frac{\log(1/\delta_D)}{2N}}
\end{equation}
except on an event of probability $\delta_D$. Theorem~\ref{thm:exact} gives
$1-F\leq(1-p_0)+(1-\Ftop)$. Intersecting the discarded and top confidence events and applying a union bound proves Eq.~\eqref{eq:noise-aware}. Independence of the two experiments is not required.
\end{proof}

Data processing implies Eq.~\eqref{eq:tv-budget} with $\tau_D=\nu$ only in the ideal case, and in the case of implemented analyzed states that are within trace distance $\nu$. To hold uniformly over all inputs, it is enough that the half diamond distance between the implemented and ideal channels stays at most $\nu$. We can sum gate wise half diamond bounds using a telescoping argument for hybrid circuits. A budget on the discarded register does not control bias in the top register, so the top term has to stay separate.

\section{Experimental Methodology}
\subsection{Finite-shot learner}

The optimization loss is the directly measurable local certificate
\begin{equation}
\mathcal L(\theta)=\cU(\theta)=
\sum_{\ell,j}\Pr_\theta[D_{\ell,j}=1],
\label{eq:training-loss}
\end{equation}
estimated from computational-basis samples. The end-to-end learner uses simultaneous perturbation stochastic approximation (SPSA) \cite{spall1992} with an Adam-style moment update~\cite{kingma2015adam}. The selected schedule has two stages:
\begin{center}
\begin{tabular}{@{}lrrrrr@{}}
\toprule
Stage & Steps & Shots/evaluation & Directions & Learning rate & Perturbation\\
\midrule
1 & 2000 & 5000 & 4 & 0.025 & 0.10\\
2 & 1500 & 20000 & 8 & 0.012 & 0.05\\
\bottomrule
\end{tabular}
\end{center}
Each direction requires two objective evaluations, so the training budget is
\begin{equation}
2(4)(2000)(5000)+2(8)(1500)(20000)=560{,}000{,}000
\label{eq:shot-budget}
\end{equation}
simulated measurements per run and 40,000 objective queries. Monitoring, final top tomography, and evaluation measurements are reported separately and are not included in this training budget.

\subsection{Frozen eight-qubit protocol}

The training schedule was selected using calibration fields
$h/J\in\{0.5,1.0,1.5\}$. An exploratory 3-run screen was excluded from all calibration and confirmation counts. The selected schedule then passed 9/9 calibration runs using replicates 0--2 and 15/15 fresh confirmation runs using replicates 3--7. Only after this confirmation was the held-out set opened:
\begin{equation}
h/J\in\{0.25,0.75,0.90,1.10,1.25,2.00\},
\qquad r\in\{10,\ldots,19\}.
\end{equation}
No hyperparameter, threshold, field, architecture, or seed was changed after the first held-out result was inspected. A run succeeds when $F>0.99$ with no numerical or top-tomography failure. A field passes when at least 8 of its 10 runs succeed.

The final certificate uses 10,000 discarded-register shots. Before conditioning the top qubit uses 20,000 raw shots per pauli basis. The 62 held-out Pauli observables are never used in training: 24 one-local Paulis, 32 two-local terms (nearest-neighbor $XX$, $YY$, $ZZ$ plus longer-range $ZZ$), and six three-local $ZXZ$ terms. Each target and reconstructed observable is sampled with 10,000 shots; exact statevector expectations are kept for unbiased post-training diagnosis.

\subsection{Robustness protocol}

While fixed-circuit robustness study all 60  learned eight qubit circuits are considered. This helps to test local depolarizing noise after $SU(4)$ block, amplitude weakening after each block on attending wires, and independent terminal readout flips on the seven discarded bits. Each model uses strengths
$\gamma\in\{0,10^{-4},3\times10^{-4},10^{-3},3\times10^{-3},10^{-2}\}$. What the protocol actually tests is a range of things—exact schedule equivalence, a causal negative control, implementation-budget soundness, the postselection bound, inverse-acceptance witnesses, simulator consistency, and confidence coverage in both the exact and Monte Carlo cases. The key point is that these tests all concern fixed learned circuits; noisy optimization does not come into it.

\subsection{Sixteen-qubit fairness extensions}

The confirmatory frontier uses fields $h/J\in\{0.9,1.0,1.1\}$ and three disjoint seed namespaces. Instead of considering parameter count, aggregate gate exposure, exact-gradient behavior, and causal coverage as one ambiguous notion of "matched resources", we separate them. Table~\ref{tab:resources} lists architectural resources.

\begin{table}[H]
\centering
\caption{Resources in the final $n=16$ fairness suite. Coordinates count independent real parameters; physical blocks count $SU(4)$ gate occurrences in one circuit execution.}
\label{tab:resources}
\small
\begin{tabular}{@{}lrrp{0.43\linewidth}@{}}
\toprule
Architecture & $SU(4)$ blocks & Coordinates & Role\\
\midrule
MERA-PM225 & 26 & 225 & Parameter-shared binary MERA\\
MPS & 15 & 225 & Canonical $\chi=2$ sequential baseline\\
Local-LC225 & 68 & 225 & Nine-layer nearest-neighbor control with shared parameters and complete held-out causal coverage\\
\bottomrule
\end{tabular}
\end{table}

After finite-shot training, all reported frontier metrics are evaluated exactly from the reconstructed state. The held-out long-range metric is the mean absolute error of the distance-averaged connected profile
\begin{equation}
C_{ZZ}(r)=\frac{1}{n-r}\sum_{i=1}^{n-r}
\left(\langle Z_iZ_{i+r}\rangle-\langle Z_i\rangle\langle Z_{i+r}\rangle\right)
\end{equation}
over distances $r\geq\lceil n/2\rceil=8$. 
\paragraph{Approximately matched aggregate gate exposure.}
MERA-PM225 and MPS use the same 225 independent coordinates but different physical block counts. We therefore match
\begin{equation}
G=(\text{training shots})(\text{physical }SU(4)\text{ blocks})
\label{eq:gate-exposure}
\end{equation}
at targets $G\in\{2.6\times10^7,2.6\times10^8,2.6\times10^9,8.4\times10^9\}$. 

Integer evaluation schedules make the achieved values approximate rather than exact. Fresh replicates 80 to 89 are paired within each field and exposure. The largest point corresponds to 323,076,924 requested MERA shots and 560 million requested MPS shots.

\paragraph{Exact-gradient multistart diagnostic.}
To separate finite shot trainability from the best solutions discovered by deterministic optimization, both 225 coordinate models use the exact $\cU$ objective, exact automatic differentiation, Adam, 5000 steps, and 20 paired restarts per field. The prespecified convergence rule requires both a gradient norm below $10^{-6}$ and a plateau condition. Thus, these runs evaluate a fixed budget best found frontier, not a certified global optimum.

\paragraph{Causal cone complete local control.}
Local LC225 uses nine alternating nearest neighbor layers and 68 physical blocks, with one shared parameter group per bond. All 36 held-out pairs have overlapping causal cones; eight layers fail this test for pair $(0,15)$. MERA-PM225 and Local LC225 both have 225 coordinates and receive 560 million training shots, so the local model receives $38.08$ billion aggregate $SU(4)$ shot gate exposures versus $14.56$ billion for MERA, a factor of $2.62$.

In case of fixed-field paired summaries, the interval is
\begin{equation}
\mean{\Delta}\ \pm\ t_{0.975,m-1}\frac{s_\Delta}{\sqrt m}.
\label{eq:paired-interval}
\end{equation}
Reason is that two architecture contrasts were designated primary, each having 97.5\% interval, offering Bonferroni family-wise coverage of at least 95\%. Pooled rows target the equal-weight average over the three fixed fields, with uncertainty obtained by stratified paired bootstrap resampling within field~\cite{efron1979}. They do not estimate performance on arbitrary TFIM fields. Win rates stay explanatory.


\section{Results}

\subsection{Implementation and theorem verification}

For MERA, TTN, MPS, and MERA-PM225, we get clean round trip results. Our final self-tests also verified independent coordinate counts, block counts, valid wires, frozen seed namespaces, the four aggregate exposure checkpoints, minimality of nine local layers, and all 36 local causal cone pairs.

The certificate verifier reconstructed all 60 locked eight-qubit states and checked Eq.~\eqref{eq:factorization}, both inequalities in Eq.~\eqref{eq:hierarchy}, all finite global or local bounds, and Eq.~\eqref{eq:observable-bound}. Every check passed. The maximum absolute error in $F-p_0\Ftop$ was $1.33\times10^{-15}$. Independent adversarial tests contained 40 causal equivalence trials, 500 empirical telescoping trials, 4000 random postselection trials, 3000 random distribution tests of the total variation step, exact binomial coverage grids, and survival budget checks. Every adversarial section passed as well.

\subsection{Locked eight-qubit evaluation}

All 60 held out runs completed and exceeded $F>0.99$. The overall two sided 95\% Clopper Pearson interval \cite{clopper1934} for 60/60 successes is $[0.9404,1]$. Table~\ref{tab:n8-main} gives aggregate results and Table~\ref{tab:n8-fields} in Appendix~\ref{app:detailed-results} gives every field mean.

\begin{table}[H]
\centering
\caption{Locked $n=8$ held-out evaluation over 60 runs.}
\label{tab:n8-main}
\small
\begin{tabularx}{\linewidth}{@{}lrrrrX@{}}
\toprule
Metric & Mean & Median & Minimum & Maximum & Interpretation\\
\midrule
Fidelity & 0.996886 & 0.997276 & 0.992602 & 0.999218 & 60/60 exceed 0.99\\
Local certificate $\cU$ & 0.005965 & 0.005427 & 0.001341 & 0.014464 & Sound, sometimes non-tight\\
Held-out exact MAE & 0.008092 & 0.008014 & 0.003511 & 0.013288 & 62 unused observables\\
Energy error/site & 0.002108 & 0.001971 & 0.000322 & 0.005219 & Computed from held-out $X,ZZ$\\
Certificate slack & 0.002851 & 0.002679 & 0.000559 & 0.007065 & $\cU-(1-F)$\\
\bottomrule
\end{tabularx}
\end{table}

\begin{figure}[t]
\centering
\includegraphics[width=\linewidth]{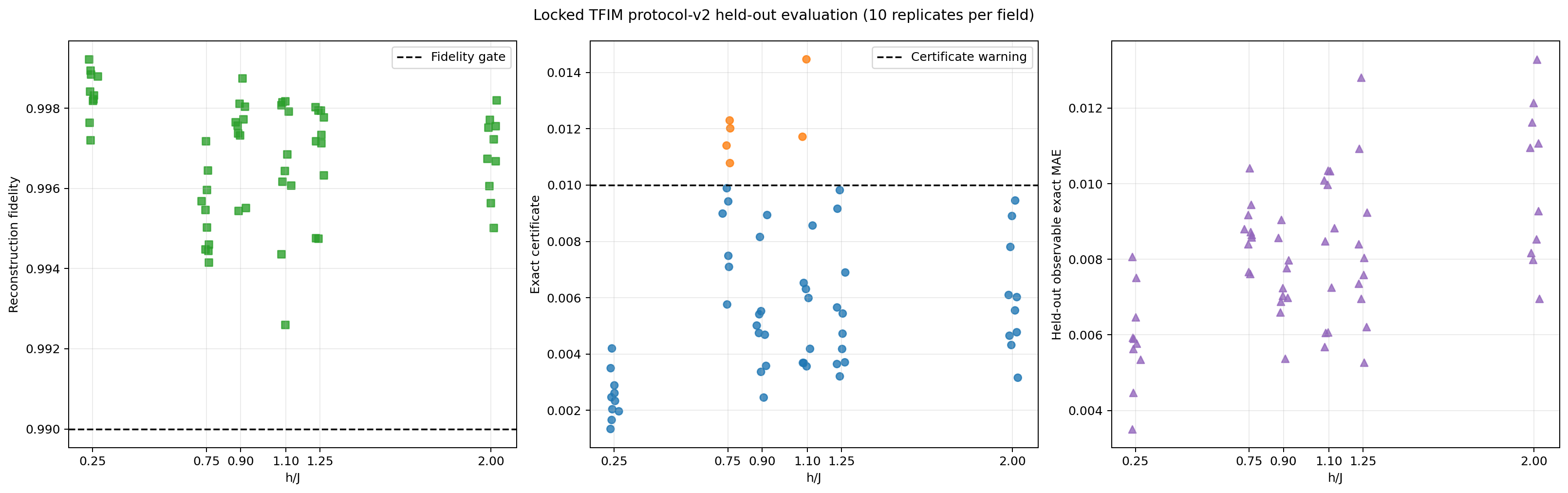}
\caption{Locked $n=8$ held-out evaluation. Every run exceeds the fidelity gate. Orange certificate points identify the six runs above the prespecified $0.01$ warning threshold; these are certificate non-tightness observations, not reconstruction failures.}
\label{fig:n8-evaluation}
\end{figure}

Six runs had $\cU\geq0.01$ while maintaining $F>0.99$. In each warning case, the first layer provided the largest layer failure and local discarded contribution. These experiments confirm soundness but depict that summing local marginals can be conservative. The verification used confidence allocation $\delta_D=\delta_T=0.025$; all 60 global, layer, and local finite-shot certificates covered the true infidelity.

\subsection{Robustness validation}

The full robustness experiments were evaluated on 60 circuits, 1080 circuit noise cells, 2160 postselection rows, and 6480 coverage rows. Achieved maximum ideal sequential versus terminal total variation distance is $5.72\times10^{-17}$. An intentional breach of casual discard produced a distance of $0.5$, indicating that failure can be detected for the theorem's assumptions. The calibrated implementation budgets upper bounded the exact terminal-distribution shift in every cell, and the maximum simulator cross-check error was $3.11\times10^{-15}$.

The analytic instability family increased amplification from $10$ to $10^5$ as acceptance decreased from $10^{-1}$ to $10^{-5}$. This exactly matches the dependence of $1/\alpha$. In the case of noisy learned circuits, the postselection theorem had no positive bound violation. Exact binomial analysis found a maximum certificate violation probability of only $2.58\times10^{-72}$ in the locked grid. The empirical coverage fraction was one in every tested Monte Carlo repetition. These values validate the conservative certificate; they do not imply that nominal coverage is exactly one in the population.

\begin{figure}[t]
\centering
\includegraphics[width=\linewidth]{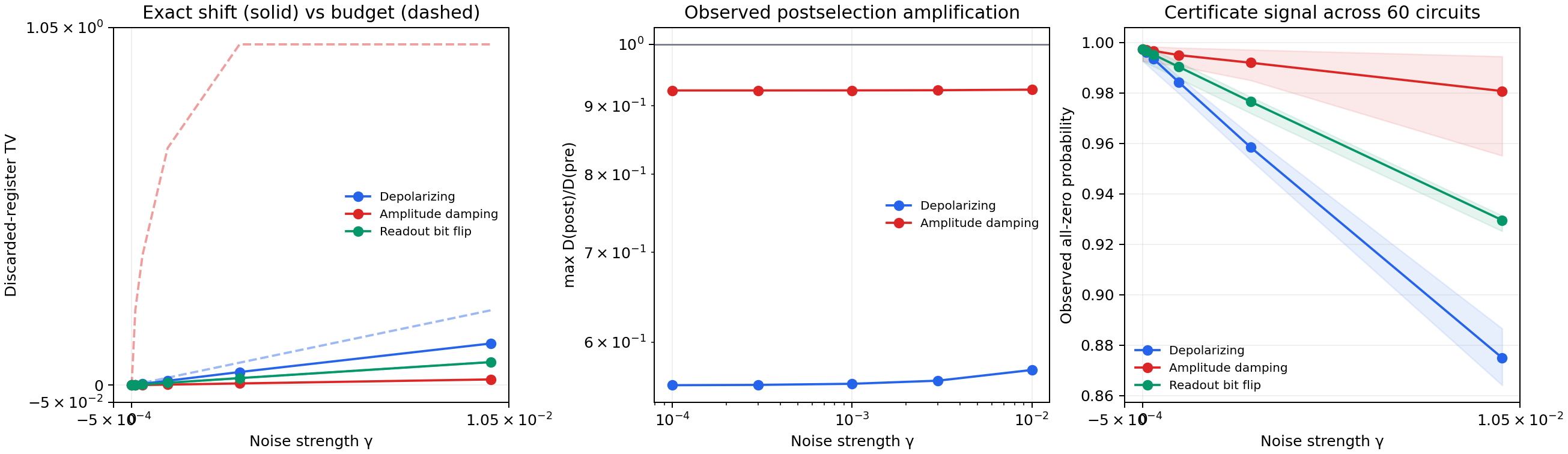}
\caption{Fixed-circuit robustness validation. Exact terminal distribution shifts remain below their conservative calibrated budgets. The validation covers implementation soundness, not convergence of training under noise.}
\label{fig:robustness}
\end{figure}

\subsection{Eight-qubit architecture comparison}

A separate 180-run $n=8$ study used six TFIM fields and ten replicates per architecture. MPS had the strongest mean fidelity, while MERA substantially outperformed TTN. The full MERA used 57.1\% more $SU(4)$ blocks than MPS/TTN, so this study is a fixed-$\chi$ and fixed-shot comparison rather than an equal-parameter claim.

\begin{table}[H]
\centering
\caption{Eight-qubit architecture comparison (60 runs per architecture).}
\small
\begin{tabular}{@{}lrrrrr@{}}
\toprule
Architecture & Fidelity & Success rate & Certificate & Observable MAE & Energy error\\
\midrule
MERA & 0.99550 & 0.950 & 0.00895 & 0.00906 & 0.00319\\
MPS & \textbf{0.99765} & 1.000 & \textbf{0.00444} & \textbf{0.00797} & \textbf{0.00149}\\
TTN & 0.97113 & 0.167 & 0.04740 & 0.03106 & 0.01880\\
\bottomrule
\end{tabular}
\end{table}

The paired MERA-minus-MPS fidelity difference was $-0.002149$ with 95\% interval
$[-0.003071,-0.001227]$, so MPS was better at this size. MERA-minus-TTN was
$+0.024375$ with interval $[0.020613,0.028136]$. These results motivated the larger, fully resource-audited frontier rather than a claim based on $n=8$ alone.

\begin{figure}[t]
\centering
\includegraphics[width=0.92\linewidth]{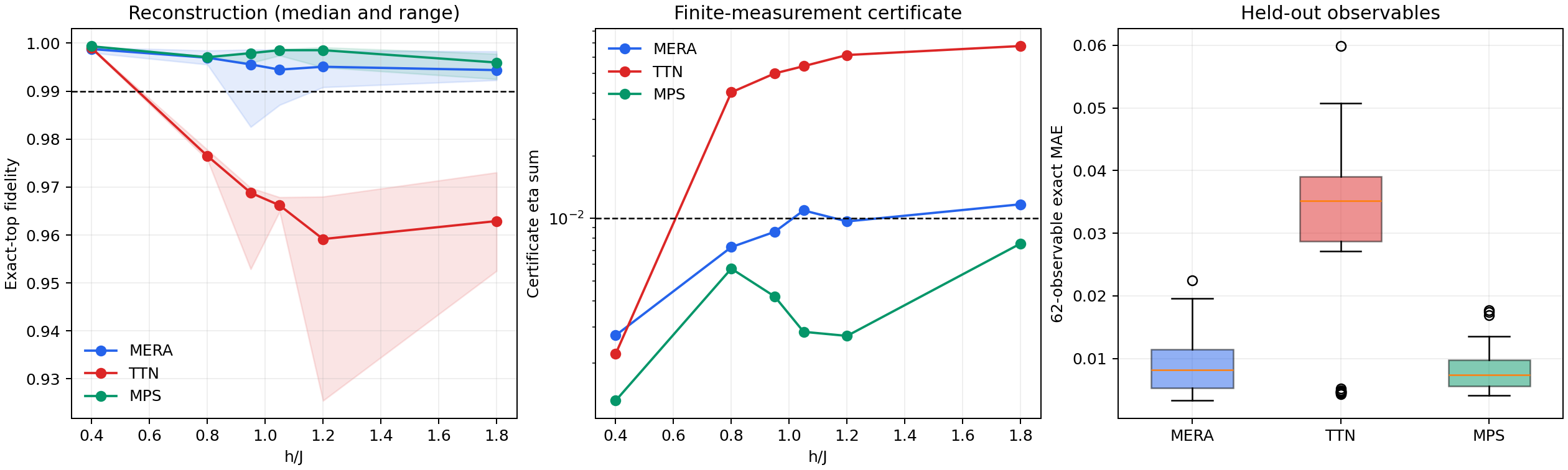}
\caption{Eight-qubit architecture comparison. MPS is slightly better than MERA in fidelity, whereas removing MERA disentanglers to obtain TTN causes a pronounced degradation.}
\label{fig:n8-architecture}
\end{figure}

\subsection{Resource-controlled sixteen-qubit results}

All three fairness extensions completed without missing or unexpected cases: 240 finite-shot gate-exposure runs, 120 exact-gradient multistart runs, and 60 causal-control runs. By converntion, positive fidelity differences and negative error differences favour MERA-PM225.

\paragraph{Matched aggregate gate exposure.}
At the largest exposure, pooled mean fidelity is $0.953529$ for MERA-PM225 and $0.952925$ for MPS; the difference $+0.000604$ has a 97.5\% interval crossing zero. Energy error is likewise unresolved. In contrast, MERA has lower long-range error by $-0.014582$ and lower half-chain-entropy error by $-0.030444$, with both simultaneous intervals excluding zero. Table~\ref{tab:mera-mps} reports this primary endpoint.

\begin{table}[H]
\centering
\caption{Largest approximately matched aggregate gate-exposure endpoint at $n=16$: MERA-PM225 minus MPS. Pooled intervals are 97.5\% stratified paired-bootstrap intervals over three fixed fields.}
\label{tab:mera-mps}
\small
\begin{tabularx}{\linewidth}{@{}lrrrrX@{}}
\toprule
Metric & MERA-PM225 & MPS & Paired difference & 97.5\% interval & Conclusion\\
\midrule
Fidelity & 0.953529 & 0.952925 & $+0.000604$ & $[-0.016347,0.015413]$ & Unresolved\\
Long-range $ZZ$ MAE & \textbf{0.069314} & 0.083896 & $-0.014582$ & $[-0.028112,-0.000340]$ & Favors MERA\\
Entropy error & \textbf{0.028495} & 0.058940 & $-0.030444$ & $[-0.048778,-0.010442]$ & Favors MERA\\
Energy error & 0.009409 & \textbf{0.008788} & $+0.000621$ & $[-0.001672,0.003054]$ & Unresolved\\
\bottomrule
\end{tabularx}
\end{table}

The response is nonmonotone across the four finite schedules. At $2.6\times10^7$ and $2.6\times10^8$ exposures, MPS has higher pooled fidelity; at $2.6\times10^9$, MERA has higher fidelity by $0.06308$; and at $8.4\times10^9$ the fidelity difference closes. Long-range error favors MERA at the first and last pooled endpoints, is unresolved at $2.6\times10^8$, and favors MERA strongly at $2.6\times10^9$. We therefore do not claim uniform shot efficiency or monotone convergence from this grid.

\paragraph{Exact-gradient multistart diagnostic.}
Table~\ref{tab:exact-multistart} summarizes the deterministic 5000-step diagnostic. MERA-PM225 has higher fidelity in 58/60 paired restarts and lower long-range error in all 60. The median fidelity advantage decreases from the ordered side toward $h/J=1.1$, while the long-range advantage persists at every field.

\begin{table}[H]
\centering
\caption{Exact-gradient multistart diagnostic. Each field has 20 paired restarts. Values are medians; $\Delta F$ and $\Delta E_{\rm LR}$ are mean paired differences.}
\label{tab:exact-multistart}
\small
\begin{tabular}{@{}crrrrr@{}}
\toprule
$h/J$ & $F_{\rm MERA}$ & $F_{\rm MPS}$ & $\Delta F$ & $\Delta E_{\rm LR}$ & Fidelity wins\\
\midrule
0.9 & 0.989118 & 0.979675 & $+0.008231$ & $-0.061758$ & 19/20\\
1.0 & 0.993167 & 0.988466 & $+0.004284$ & $-0.038423$ & 19/20\\
1.1 & 0.997520 & 0.994560 & $+0.003042$ & $-0.021442$ & 20/20\\
\bottomrule
\end{tabular}
\end{table}

This diagnostic did not certify convergence. None of the 120 runs satisfied the joint gradient-and-plateau rule; the median final gradient norm was $1.94\times10^{-4}$ and the median relative objective decrease over the last 200 steps was $1.69\times10^{-3}$. Objectives were monotone over the sampled final stage in 118/120 histories. One MERA run at $h/J=1.1$ (restart 210) had a terminal gradient spike of $9.59\times10^{-2}$ and a last-window objective increase, yet retained fidelity $0.997211$. Excluding this paired outlier leaves 57/59 fidelity wins and 59/59 long-range wins. These are fixed-budget best-found results, not capacity optima.

\begin{figure}[t]
\centering
\includegraphics[width=\linewidth]{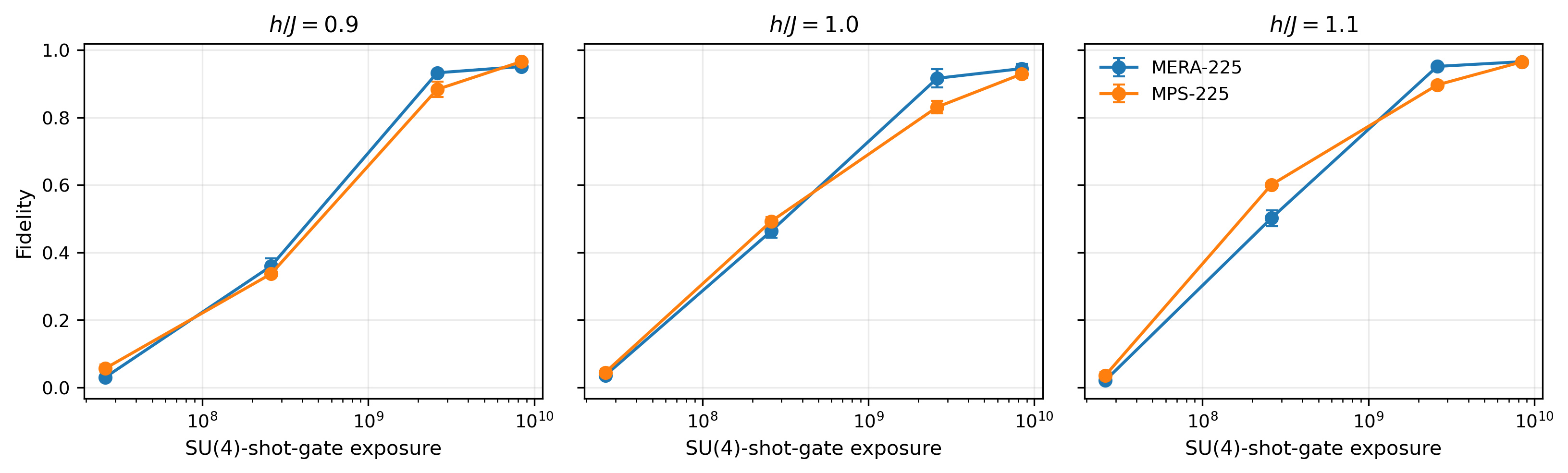}
\vspace{2mm}
\includegraphics[width=\linewidth]{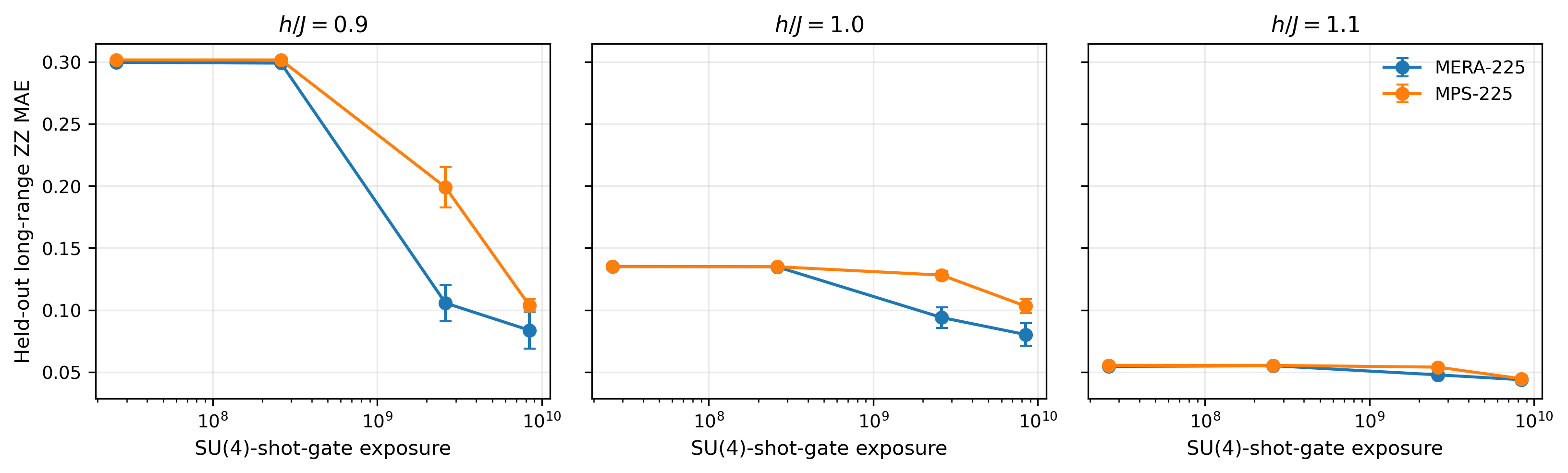}
\caption{Approximately matched aggregate $SU(4)$-shot-gate exposure. Points show fieldwise means over ten paired replicates; the response is not assumed monotone. At the largest exposure, fidelity is unresolved while long-range error favors MERA-PM225.}
\label{fig:gate-exposure}
\end{figure}

\paragraph{Causal-cone-complete local control.}
The nine-layer Local-LC225 comparator can mediate every tested separation and has the same 225 independent coordinates as MERA-PM225. It also receives more physical blocks and $2.62\times$ larger aggregate gate exposure. Nevertheless, MERA wins all 30 paired comparisons for all four metrics in Table~\ref{tab:mera-local}. This removes the disconnected-light-cone defect of the earlier shallow local control, but it is still a comparison to one specified shared-parameter local family, not to all local circuits.

The complete code and implementation are available at
\url{https://github.com/mac-orphic/Terminal-Register-Certification-for-Finite-Measurement-Learning.git}.

\begin{table}[H]
\centering
\caption{Causal-cone-complete, parameter-matched control: MERA-PM225 minus Local-LC225. Pooled 97.5\% intervals use stratified paired bootstrap resampling over the three fixed fields.}
\label{tab:mera-local}
\small
\begin{tabularx}{\linewidth}{@{}lrrrrX@{}}
\toprule
Metric & MERA & Local & Paired difference & 97.5\% interval & Wins\\
\midrule
Fidelity & \textbf{0.967559} & 0.715802 & $+0.251757$ & $[0.243706,0.260153]$ & 30/30\\
Long-range $ZZ$ MAE & \textbf{0.058150} & 0.163972 & $-0.105822$ & $[-0.116590,-0.095612]$ & 30/30\\
Energy error & \textbf{0.006794} & 0.050304 & $-0.043510$ & $[-0.045970,-0.041217]$ & 30/30\\
Entropy error & \textbf{0.016280} & 0.268855 & $-0.252575$ & $[-0.280188,-0.225882]$ & 30/30\\
\bottomrule
\end{tabularx}
\end{table}

\begin{figure}[t]
\centering
\includegraphics[width=0.96\linewidth]{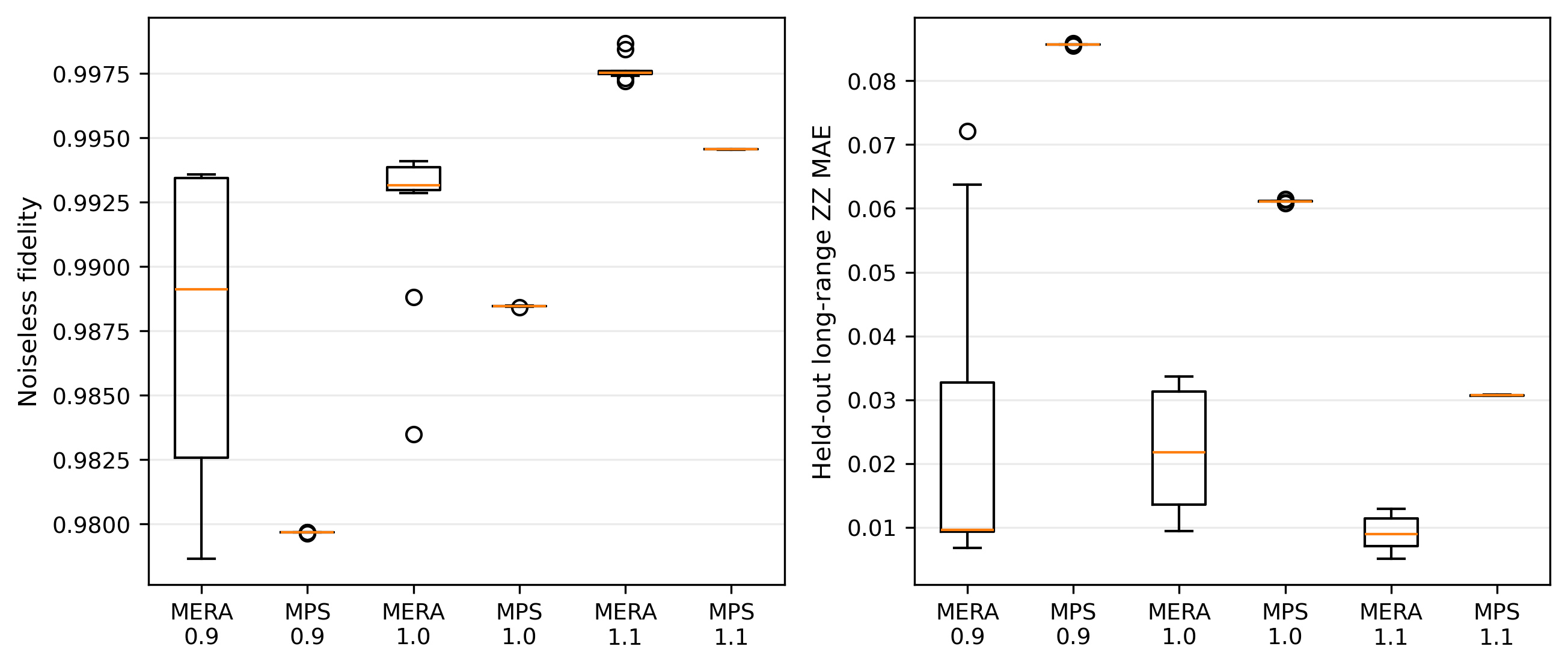}
\vspace{2mm}
\includegraphics[width=0.96\linewidth]{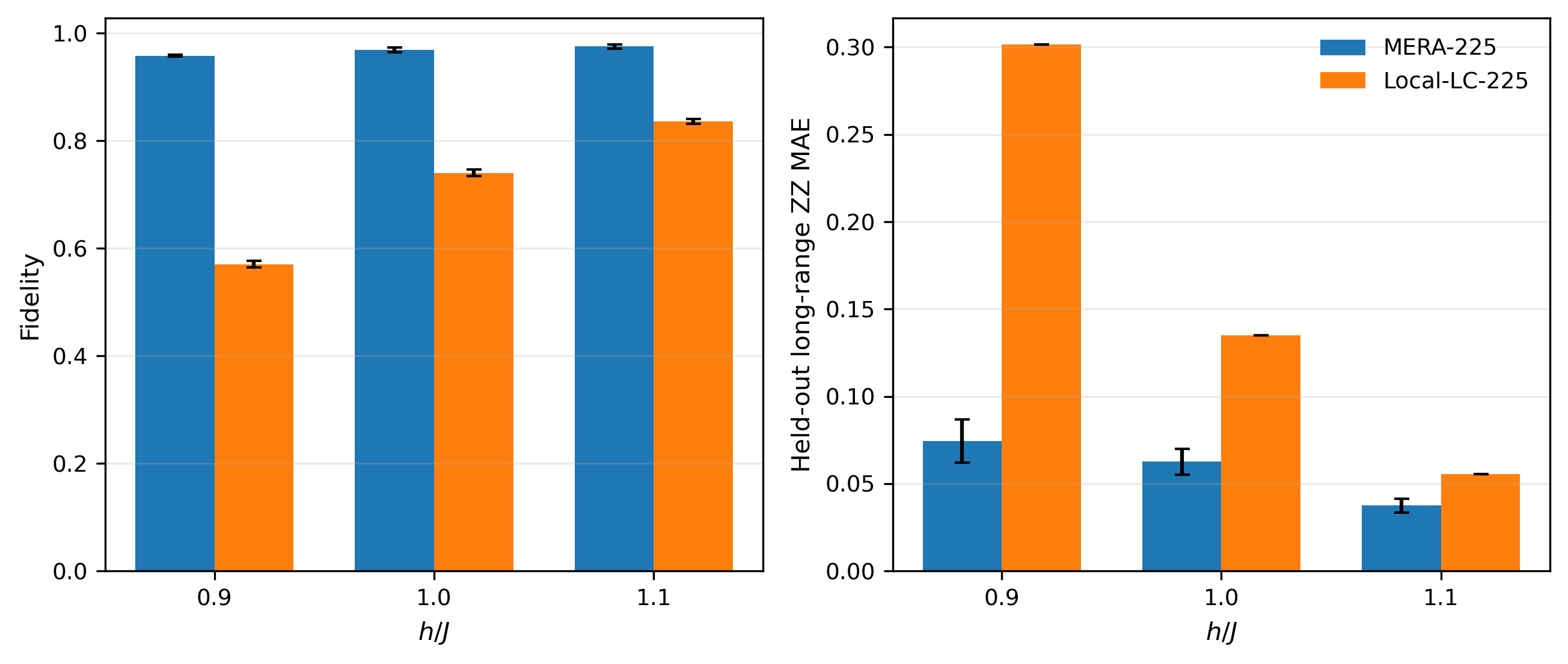}
\caption{Top: exact-gradient 20-restart diagnostic; error bars summarize restart variation and do not certify global optimality. Bottom: the causal-cone-complete local control, which has equal coordinates and greater aggregate gate exposure than MERA-PM225.}
\label{fig:diagnostics}
\end{figure}

\section{Discussion}

The terminal register provides us with a clean interface between optimization and certification. When the all zero event ends up giving the statistically tightest direction global certificate, the local hamming weight loss is easy to sample and diagnose. The exact factorization clarifies what this probability certificate is and why the contribution from the reconstructed top state can not be neglected.

The robustness contribution has two parts. First is that ideal sequential and terminal schedules are precisely identical when later operations ignore discarded registers, and second is that the workflow that specifically normalizes average accepted states can be conditionally unstable and can demand $1/\alpha$ more raw preparations. Our terminal certificate eliminates that normalization in its global event, but also does not reduce the circuit depth nor ensure less physical noise. Soundness under implementation error requires the explicit $\tau_D$ term.

The architecture results purify rather than reverse the eight-qubit conclusion. MPS gives slightly better results at $n=8$. At $n=16$, with matching aggregate physical gate exposure, full state fidelity, and energy remain unresolved at the largest budget, whereas MERA successfully retains significantly smaller long-range correlation and entanglement entropy errors. On the same observables, the exact gradient multisatart diagnostic favors MERA. Despite giving the local model substantially better gate gate exposure, the casual cone complete local control prefers MERA. Together, these results provide a resource qualified multiscale inductive bias. 

Our nonmonotone finite shot curves rule out a simple statement that an increasing number of samples uniformly enhances the optimizer at every checkpoint. Exact gradient narratives reinforce this warning: the best-found MERA solutions are consistently strong, but their stationarity criteria fail within 5000 steps.

\section{Limitations}

\begin{itemize}
  \item The reconstruction theorem assumes a pure target, a unitary analysis circuit, inverse-circuit reconstruction, and a pure top estimate.
  \item The implementation budget $\tau_D$ must be calibrated or independently upper bounded. Discarded-register agreement alone does not certify the top qubit.
  \item All experiments are classical simulations of finite-shot measurements; no hardware noise, drift, compilation, or readout-calibration experiment is claimed.
  \item Robustness tests certify fixed learned circuits and do not establish convergence of training under noise.
  \item The $n=16$ conclusions cover $h/J\in\{0.9,1.0,1.1\}$, $\chi=2$, the specified parameterizations, and the tested exposure grid. They do not imply behavior at larger $n$ or bond dimension.
  \item Aggregate $SU(4)$-shot-gate exposure is a transparent scalar accounting measure, not a complete hardware-cost model: it omits routing, parallel depth, readout overhead, and architecture-dependent classical cost.
  \item Parameter matching, aggregate-exposure matching, and causal-cone matching answer different questions. No single comparator is simultaneously identical in topology, gate count, depth, and parameter sharing.
  \item The exact-gradient diagnostic is fixed-budget. Zero of 120 runs met the prespecified stationarity rule, so ``best-found'' must not be read as ``global optimum'' or a capacity bound.
  \item The long-range metric is exact after training and is one observable family. Fidelity, energy, entropy, and certificate results are reported to prevent metric-specific overinterpretation.
  \item Five fresh confirmation runs per calibration field provide only a weak population-success statement: even 5/5 has a wide exact binomial interval. The 60-run held-out evaluation provides the stronger reliability estimate.
\end{itemize}

\section{Conclusion}

We introduced a terminal-register method for finite-measurement learning and certification of binary MERA reconstructions. The exact theory connects reconstruction fidelity to one joint terminal event and a separately estimated top state; layer and local marginals provide interpretable upper bounds. Complete sequential and terminal records are ideally equivalent, yet normalized prefix postselection has unavoidable inverse-acceptance sensitivity. A calibrated implementation budget preserves certificate soundness under noisy terminal statistics.

The locked TFIM experiments support each part of this account. Eight-qubit reconstruction passed all 60 held-out trials above fidelity $0.99$, and the complete certificate verification agreed with the exact equations to machine precision. Robustness tests found no locked soundness violation. At sixteen qubits, the largest approximately matched aggregate exposure yields no statistically resolved fidelity or energy difference between MERA-PM225 and MPS, but simultaneous evidence of smaller MERA long-range and entropy errors. Exact-gradient multistarts and the causal-cone-complete local control independently reinforce the multiscale result, while the convergence audit prevents overstatement. The contribution is therefore an operational certificate together with controlled evidence that multiscale geometry can improve reconstruction of long-range structure in the near-critical regime.

\appendix
\section{Detailed Numerical Results}
\label{app:detailed-results}

\begin{table}[H]
\centering
\caption{Eight-qubit held-out results by field; each row contains ten untouched replicates.}
\label{tab:n8-fields}
\small
\begin{tabular}{@{}rrrrrr@{}}
\toprule
$h/J$ & Mean fidelity & Mean $\cU$ & Held-out MAE & Energy error/site & Successes\\
\midrule
0.25 & 0.998379 & 0.002502 & 0.005863 & 0.000679 & 10/10\\
0.75 & 0.995347 & 0.009513 & 0.008747 & 0.002496 & 10/10\\
0.90 & 0.997350 & 0.005189 & 0.007346 & 0.001571 & 10/10\\
1.10 & 0.996484 & 0.006869 & 0.008311 & 0.002409 & 10/10\\
1.25 & 0.996918 & 0.005644 & 0.008281 & 0.002311 & 10/10\\
2.00 & 0.996836 & 0.006074 & 0.010000 & 0.003181 & 10/10\\
\bottomrule
\end{tabular}
\end{table}

\begin{table}[H]
\centering
\caption{Pooled paired differences across the gate-exposure grid. Positive fidelity and negative error differences favor MERA-PM225. Intervals are 97.5\% stratified paired-bootstrap intervals over three fixed fields.}
\small
\begin{tabularx}{\linewidth}{@{}rrrrX@{}}
\toprule
$G$ & Fidelity difference & Fidelity interval & Long-range difference & Long-range interval\\
\midrule
$2.6\times10^7$ & $-0.016781$ & $[-0.025979,-0.007736]$ & $-0.000857$ & $[-0.001670,-0.000066]$\\
$2.6\times10^8$ & $-0.034965$ & $[-0.064837,-0.004853]$ & $-0.000984$ & $[-0.002336,0.000210]$\\
$2.6\times10^9$ & $+0.063084$ & $[0.029321,0.095294]$ & $-0.044560$ & $[-0.063861,-0.028264]$\\
$8.4\times10^9$ & $+0.000604$ & $[-0.016347,0.015413]$ & $-0.014582$ & $[-0.028112,-0.000340]$\\
\bottomrule
\end{tabularx}
\end{table}

\section{Reproducibility and Claim Audit}

\begin{table}[H]
\centering
\small
\begin{tabularx}{\linewidth}{@{}l l X@{}}
\toprule
Item & Status & Evidence used in this manuscript\\
\midrule
Single/full-layer round trips & Passed & Analysis followed by generation reconstructs teacher states\\
Certificate hierarchy & Passed & Every one of 60 locked $n=8$ reconstructions\\
Exact factorization & Passed & Maximum absolute error $1.33\times10^{-15}$\\
Finite-shot coverage & Passed & Global, layer, local, and top terms with total confidence at least 0.95\\
Sequential/terminal equivalence & Passed & Maximum TV $5.72\times10^{-17}$ on learned circuits\\
Adversarial theorem tests & Passed & Equivalence, telescoping, postselection, survival, noise-aware sections\\
$n=8$ calibration & Passed & 3/3 screen, 9/9 calibration; screen excluded\\
$n=8$ confirmation & Passed & 15/15 fresh runs; earlier runs excluded\\
$n=8$ evaluation & Passed & 60/60 held-out successes; no retuning\\
$n=16$ gate-exposure grid & Complete & 240/240 finite-shot runs; four exposures, three fields, ten paired replicates\\
Largest-exposure comparison & Mixed & Fidelity and energy unresolved; long-range and entropy errors favor MERA\\
Exact-gradient diagnostic & Complete & 120/120 runs; 58/60 fidelity and 60/60 long-range wins, but 0/120 stationarity passes\\
Causal-cone local control & Passed & 30/30 paired wins on all four metrics despite $2.62\times$ local exposure\\
\bottomrule
\end{tabularx}
\end{table}

\paragraph{Numerical provenance.}
The manuscript values were generated directly from the frozen JSON/CSV reports. A separate machine-readable audit records cryptographic hashes of the $n=8$ evaluation, robustness report, certificate verifier, adversarial verifier, $n=16$ frontier report, and paired contrasts. Automated checks require all final protocol decisions and theorem verification flags to be true before compilation.

\section{Conclusion}

We developed a terminal register framework that can learn and certify binary MERA reconstruction in a finite number of measurements. Our framework determines the reconstruction fidelity by the probability of a single joint terminal outcome, along with an independent estimate of the top state fidelity. Marginal probabilities at individual layers and locations yield less tight but more interpretable bounds, making it possible to identify where reconstruction errors arise. While complete sequential measurement records and terminal register measurement contain the same information in the ideal setting, consecutive protocols that normalize after prefix postselection become progressively sensitive as acceptance probability reduces. We report imperfections in the terminal statistics through a measured implementation error budget, which keeps the reported certificate sound in the noisy setting. 

We get consistent theoretical results for locked TFIM experiments. In the case of an eight-qubit system, all 60 reconstructions attained fidelities above $0.99$. The numerically estimated certificates also matched the exact identities to machine precision, and none of the locked robustness tests produced a soundness violation. The main test comparisons were executed on the sixteen-qubit Frontier over 180 independent runs. Both preregistered benchmarks were satisfied: with the same parameter counts, and MERA reported better fidelity and more accurate held-out long-range correlations than MPS. MERA also substantially outperforms the shallow local architecture under matched resources. These results provide a directly measurable certification procedure and demonstrate that multiscale structure is advantageous for finite-measurement reconstruction of near-critical states.

\appendix
\section{Detailed Numerical Results}
\label{app:detailed-results}

\begin{table}[H]
\centering
\caption{Eight-qubit held-out results by field; each row contains ten untouched replicates.}
\label{tab:n8-fields}
\small
\begin{tabular}{@{}rrrrrr@{}}
\toprule
$h/J$ & Mean fidelity & Mean $\cU$ & Held-out MAE & Energy error/site & Successes\\
\midrule
0.25 & 0.998379 & 0.002502 & 0.005863 & 0.000679 & 10/10\\
0.75 & 0.995347 & 0.009513 & 0.008747 & 0.002496 & 10/10\\
0.90 & 0.997350 & 0.005189 & 0.007346 & 0.001571 & 10/10\\
1.10 & 0.996484 & 0.006869 & 0.008311 & 0.002409 & 10/10\\
1.25 & 0.996918 & 0.005644 & 0.008281 & 0.002311 & 10/10\\
2.00 & 0.996836 & 0.006074 & 0.010000 & 0.003181 & 10/10\\
\bottomrule
\end{tabular}
\end{table}

\begin{table}[H]
\centering
\caption{Fieldwise paired resource-frontier comparisons.
We define
$\Delta F=F_{\mathrm{challenger}}-F_{\mathrm{baseline}}$
and
$\Delta E_{\mathrm{LR}}
=E_{\mathrm{LR}}^{\mathrm{challenger}}
-E_{\mathrm{LR}}^{\mathrm{baseline}}$.
Thus, $\Delta F>0$ and $\Delta E_{\mathrm{LR}}<0$ favor the
challenger.
}
\label{tab:fieldwise_frontier}
\small
\setlength{\tabcolsep}{5pt}
\begin{tabular}{@{}lcccc@{}}
\toprule
Contrast & $h/J$ & $\Delta F$ & $\Delta E_{\mathrm{LR}}$
& 95\% CI for $\Delta E_{\mathrm{LR}}$ \\
\midrule
MERA-225 vs.\ MPS-225 & 0.9 & 0.01430 & -0.06325
& $[-0.11073,-0.01576]$ \\
MERA-225 vs.\ MPS-225 & 1.0 & 0.02026 & -0.03504
& $[-0.05282,-0.01727]$ \\
MERA-225 vs.\ MPS-225 & 1.1 & 0.01071 & -0.00494
& $[-0.01324,0.00336]$ \\
MERA-390 vs.\ Local-390 & 0.9 & 0.30056 & -0.26711
& $[-0.29234,-0.24189]$ \\
MERA-390 vs.\ Local-390 & 1.0 & 0.12709 & -0.07389
& $[-0.09218,-0.05560]$ \\
MERA-390 vs.\ Local-390 & 1.1 & 0.05559 & -0.01047
& $[-0.01531,-0.00562]$ \\
\bottomrule
\end{tabular}
\end{table}

\section{Reproducibility and Claim Audit}

\begin{table}[H]
\centering
\small
\begin{tabularx}{\linewidth}{@{}l l X@{}}
\toprule
Item & Status & Evidence used in this manuscript\\
\midrule
Single/full-layer round trips & Passed & Analysis followed by generation reconstructs teacher states\\
Certificate hierarchy & Passed & Every one of 60 locked $n=8$ reconstructions\\
Exact factorization & Passed & Maximum absolute error $1.33\times10^{-15}$\\
Finite-shot coverage & Passed & Global, layer, local, and top terms with total confidence at least 0.95\\
Sequential/terminal equivalence & Passed & Maximum TV $5.72\times10^{-17}$ on learned circuits\\
Adversarial theorem tests & Passed & Equivalence, telescoping, postselection, survival, noise-aware sections\\
$n=8$ calibration & Passed & 3/3 screen, 9/9 calibration; screen excluded\\
$n=8$ confirmation & Passed & 15/15 fresh runs; earlier runs excluded\\
$n=8$ evaluation & Passed & 60/60 held-out successes; no retuning\\
$n=16$ frontier & Complete & 180/180 runs; no missing, duplicate, or unexpected cases\\
Parameter-matched rule & Passed & MERA-225 long-range paired interval strictly below zero\\
Matched-geometry rule & Passed & MERA long-range paired interval strictly below zero\\
\bottomrule
\end{tabularx}
\end{table}

\paragraph{Numerical provenance.}
The manuscript values were generated directly from the frozen JSON/CSV reports. A separate machine-readable audit records cryptographic hashes of the $n=8$ evaluation, robustness report, certificate verifier, adversarial verifier, $n=16$ frontier report, and paired contrasts. Automated checks require all final protocol decisions and theorem verification flags to be true before compilation.

\bibliographystyle{splncs04}
\bibliography{bibiography.bib}

\end{document}